\documentclass[11pt]{article}
\usepackage{fullpage}
\usepackage{amsfonts}
\usepackage{algorithmic}
\usepackage{algorithm}
\usepackage{amsmath,amsthm,amssymb}
\usepackage{latexsym}
\usepackage{epic}
\usepackage{epsfig}
\usepackage{amscd}
\usepackage{url}
\usepackage{verbatim}

\newtheorem{theorem}{Theorem}[section]

\newtheorem{lemma}[theorem]{Lemma}

\newtheorem{definition}[theorem]{Definition}
\newtheorem{remark}[theorem]{Remark}

\newtheorem{example}[theorem]{Example}

\newtheorem{prop}[theorem]{Proposition}

\newcommand{\cross}{\times}

\newcommand{\set}[1]{\left\{ #1 \right\}}
\newcommand{\union}{\bigcup}
\newcommand{\intersect}{\bigcap}
\newcommand{\sm}{\setminus}

\renewcommand{\hat}{\widehat}
\renewcommand{\tilde}{\widetilde}

\def\min{\qopname\relax n{min}}
\def\max{\qopname\relax n{max}}

\def\Pr{\qopname\relax n{\mathbf{Pr}}}
\def\Ex{\qopname\relax n{\mathbf{E}}}

\def\A{\mathcal{A}}

\def\I{\mathcal{I}}
\def\J{\mathcal{J}}

\def\R{\mathcal{R}}

\def\X{\mathcal{X}}

\def\sse{\subseteq}

\newcommand{\eat}[1]{}

\newenvironment{lp*}{\begin{equation*}  \begin{array}{lll}}{\end{array}\end{equation*}}

\newcommand{\sdc}{SDC }

\title{Succinct Coverage Oracles}
\author{Ioannis Antonellis$^1$, Anish Das Sarma$^2$, Shaddin Dughmi$^1$\\
 $^1$Stanford University, $^2$Yahoo Research\\
  $^1${\tt\{antonell,shaddin\}@cs.stanford.edu}, $^2${\tt anishdas@yahoo-inc.com}\\
} 
\date{}
\addtocounter{page}{-1}

\begin{document}

\maketitle
\thispagestyle{empty}

\begin{abstract}
 
  In this paper, we identify a fundamental algorithmic problem that we
  term {\em succinct dynamic covering} (SDC), arising in many
  modern-day web applications, including ad-serving and online
  recommendation systems in eBay and Netflix.  Roughly speaking, \sdc\
  applies two restrictions to the well-studied Max-Coverage
  problem~\cite{garey}: Given an integer $k$, $\X=\{1,2,\ldots,n\}$
  and $\I=\{S_1, \ldots, S_m\}$, $S_i\subseteq \X$, find $\J\subseteq
  \I$, such that $|\J|\leq k$ and $ (\union_{S\in \J} S)$ is as large
  as possible. The two restrictions applied by \sdc\ are: (1) {\em
    Dynamic:} At query-time, we are given a \emph{query} $Q\subseteq
  \X$, and our goal is to find $\J$ such that $Q\bigcap (\union_{S\in
    \J} S)$ is as large as possible; (2) {\em Space-constrained:} We
  don't have enough space to store (and process) the entire input;
  specifically, we have $o(mn)$, and maybe as little as
  $O((m+n)polylog(mn))$ space. A solution to SDC maintains a small data
  structure, and uses this datastructure to answer {\em most} dynamic queries with high
  accuracy. We call such a scheme a \emph{Coverage Oracle}.

  We present algorithms and complexity results for coverage oracles. We present
  deterministic and probabilistic near-tight upper and lower bounds on
  the approximation ratio of \sdc as a function of the amount of
  space available to the oracle. Our lower bound results show that to obtain
  constant-factor approximations we need $\Omega(mn)$
  space. Fortunately, our upper bounds present an explicit tradeoff
  between space and approximation ratio, allowing us to determine the
  amount of space needed to guarantee certain accuracy. 
\end{abstract}
\newpage

\section{Introduction}
\label{sec:intro}

The explosion of data and applications on the web over the last decade have given rise to many new data management challenges. This paper identifies a fundamental subproblem inherent in several Web applications, including online recommendation systems, and serving advertisements on webpages. Let us begin with a motivating example.

\begin{example}
Consider the online movie rental and streaming website, Netflix\footnote{www.netflix.com}, and one of their users Alice. Based on Alice's movie viewing (and rating) history, Netflix would like to recommend new movies to Alice for watching. (Indeed, Netflix threw open a million-dollar challenge on significantly improving their movie recommendations~\footnote{http://www.netflixprize.com/}.) Conceivably, there are many ways of devising algorithms for recommendation ranging from data mining to machine learning techniques, and indeed there has been a great deal of such work on providing personalized recommendations (see~\cite{tuzhilin} for a survey). Regardless of the specific technique, an important subproblem that arises is finding users ``similar'' to Alice, i.e., finding users who have independently or in conjunction viewed (and liked) movies seen (and liked) by Alice.

Abstractly speaking, we are given a {\em universal set} of all Netflix movies, and Netflix users identified by the {\em subset} of movies they have viewed (and liked or disliked). Given a specific user Alice, we are interested in finding (say $k$) other users, who together {\em cover} a large set of Alice's likes and dislikes. Note that for each user, the set of movies that need to be covered is different, and therefore the covering cannot be performed {\em statically}, independent of the user. In fact, Netflix {\em dynamically} provides movie recommendations as users rate movies in a particular genre (say comedy), or request movies in specific languages, or time periods. Providing recommendations at interactive speed, based on user queries (such as for a particular genre), rules out computationally-expensive processing over the entire Netflix data, which is very large.\footnote{Netflix currently has over 10 millions users, over 100,000 movies, and obviously some of the popular movies have been viewed by many users, and movie buffs have rated a large number of movies; Netflix owns over 55 million discs.} Therefore, we are interested in {\em approximately} solving the aforementioned covering problem based on a {\em subset} of the data. 

The main challenge that arises is to {\em statically} identify a {\em subset} of the data that would provide good approximations to the covering problem for any {\em dynamic} user query.
\end{example}

\noindent Note that a very similar challenge arises in other recommendation systems, such as when Alice visits an online shopping website like eBay\footnote{www.ebay.com} or Amazon\footnote{www.amazon.com}, and the website is interested in recommending products to Alice based on her current query for a particular brand or product, and her prior purchasing (and viewing) history.

The example above can be formulated as an instance of a simple
algorithmic covering problem, generalizing the NP-hard optimization
problem \emph{max k-cover}~\cite{garey}. The input to this problem is
an integer $k$, a set $\X= \set{1,\ldots,n}$, a family $\I \sse 2^\X$
of subsets of $\X$, and \emph{query} $Q \sse X$. Here $(\X,\I)$ is
called a \emph{set system}, $\X$ is called the \emph{ground set} of
the set-system, and members of $\X$ are called \emph{elements} or
\emph{items}.  We make no assumptions on how the set system is
represented in the input, though the reader can think of the obvious
representation by a $n \times m$ bipartite graph for intuition. This
$n \times m$ bipartite graph can be stored in $O(nm)$ bits, which is
in fact information-theoretically optimal for storing an arbitrary set
system on $n$ items and $m$ sets. The objective of the problem is to
return $\J \sse \I$ with $|\J| \leq k$ that collectively cover as much
of $Q$ as possible. Since this problem is a generalization of max
k-cover, it is NP-hard. Nevertheless, absent any additional
constraints this problem can be approximated in polynomial time by a
straightforward adaptation of the greedy algorithm for max
k-cover \footnote{The greedy algorithm for max k-cover, adapted to our
  problem, is simple: Find the set in $\I$ covering as many uncovered
  items in $Q$ as possible, and repeat this $k$ times. This can
  clearly be implemented in $O(mn)$ time, and has been shown to yield
  a $e/(e-1)$ approximation.}, which attains a constant factor
$\frac{e}{e-1}$ approximation in $O(mn)$ time \cite{nwf78}. However, we
further constrain solutions to the problem as follows, rendering new
techniques necessary.

From the above example, we identify two properties that we require of
any system that solves this covering problem:
\begin{enumerate}
\item {\bf  Space Constrained:} We need to (statically) preprocess the
  set system $(\X,\I)$ and store a small sketch (much smaller than
  $O(mn)$), in the form of a \emph{data structure}, and discard the
  original representation of $(\X,\I)$. This can be thought of as a
  form of lossy compression. We do not require the data
  structure to take any particular form; it need only be a sequence of
  bits that allows us to extract information about the original set
  system $(\X,\I)$. For instance, any statistical summary, a subgraph
  of the bipartite graph representing the set system, or other
  representation is acceptable.
\item {\bf Dynamic}: The query $Q$ is not known a-priori, but arrive
  \emph{dynamically}. More precisely: $Q$ arrives \emph{after} the
  data structure is constructed and the original data discarded. It is
  at that point that the data structure must be used to compute a
  solution $\J$ to the covering problem.
\end{enumerate}

We call this covering problem (formalized in the next section) the
{\em Succinct Dynamic Covering} (SDC) problem. Moreover, we call a
solution to SDC a \emph{Coverage Oracle}. A coverage oracle consists
of a \emph{static stage} that constructs a datastructure, and a
\emph{dynamic stage} that uses the datastructure to answer queries.

Next we briefly present another, entirely different, Web application
that also needs to confront \sdc.  In addition, we note that there
are several other applications facing similar covering problems,
including gene identification~\cite{genomics}, searching domain-specific aggregator sites like Yelp~\footnote{www.yelp.com}, topical query decomposition~\cite{tqd}, and search-result diversification~\cite{CG98,CK06}.

\begin{example}
  Online advertisers bid on (1) webpages matching relevancy criteria
  and (2) typically target a certain user demographic. Advertisements
  are served based on a combination of the two criterion above. When a
  user visits a particular webpage, there is usually no precise
  information about the users' demographic, i.e., age, location,
  interests, gender, etc. Instead, there is a {\em range} of possible
  values for each of these attributes, determined based on the search
  query the user issued or session information. Ad-servers therefore
  attempt to pick a set of advertisements that would be of interest
  (i.e., ``cover'') a large number of users; the user demographic that
  needs to be covered is determined by the page on which the
  advertisement is being placed, the user query, and session
  information. Therefore, ad-serving is faced with the \sdc\
  problem. The space constraint arises because the set system
  consisting of all webpages, and each user identified by the set of
  webpages visited by the user is prohibitively large to store in memory and
  process in real-time for every single page view. The dynamic aspect
  arises because each user view of each page is associated with a
  different user demographic that needs to be covered.  
\end{example}

\subsection{Contributions and Outline}

Next we outline the main contributions of this paper.

\begin{itemize}

\item In Section~\ref{sec:defns} we formally define the succinct dynamic covering (SDC) problem, and summarize our results.

\item In Section~\ref{sec:general_ub} we present a randomized
  coverage oracle for \sdc. The oracle is presented as a
  function of the available space, thus allowing us to tradeoff space
  for accuracy based on the specific application. Unfortunately, the
  approximation ratio of this oracle degrades rapidly as space
  decreases; However, the next section shows that this is in fact
  unavoidable.

\item In Section  \ref{sec:general_lb} we present a lowerbound on
  the best possible approximation attainable as a function of the
  space allowed for the datastructure. This lowerbound essentially
  matches the upperbound of Section \,\ref{sec:general_ub}, though
  with the caveat that the lowerbound is for oracles that do not use
  randomization. We expect the lowerbound to hold more generally for
  randomized oracles, though we leave this as an open question.

\end{itemize}

\noindent Related work and future directions are presented in
Section~\ref{sec:conclusions}.

\subsection{Related Work}

Our study of the tradeoff between space and approximation ratio is in the spirit of the  work of Thorup and Zwick \cite{thorupzwick} on \emph{distance oracles}. They  considered the problem of compressing a graph $G$ into a
small datastructure, in such a way that the datastructure can be used
to approximately answer queries for the distance between pairs of
nodes in $G$. Similar to our results, they showed matching upper and
lower bounds on the space needed for compressing the graph subject to
preserving a certain approximation ratio. Moreover, similarly to our
upperbounds for SDC, their distance oracles benefit from a
speedup at query time as approximation ratio is sacrificed for space.

Previous work has studied the set cover problem under streaming models. One model studied in~\cite{AAASTOC2003,NBESA2005} assumes that the sets are known in advance, only elements
arrive online, and, the algorithms do not know in advance which subset
of elements will arrive. An alternative model assumes that elements
are known in advance and sets arrive in a streaming
fashion~\cite{SGSDM2009}. Our work differs from these works in that
SDC operates under a storage budget, so all sets cannot be stored;
moreover, SDC needs to provide a good cover for {\em all} possible
dynamic query inputs. 

Another related area is that of nearest neighbor search. It is easy to
see that the \sdc~problem with $k=1$ corresponds to nearest neighbor
search using the dot product similarity measure, i.e.,
$sim_{dot}(x,y) = \frac{dot(x,y)}{n}$. However, following from  a
result from Charikar~\cite{CSTOC2002}, there exists  no locality
sensitive hash function family for the dot product similarity
function. Thus, there is no hope that signature schemes (like
minhashing for the Jaccard distance) can be used for \sdc.

\section{SDC}
\label{sec:defns}

We start by defining the succinct dynamic covering (SDC) problem in Section~\ref{subsec:defn}. Then, in \eat{Section~\ref{subsec:timespeedup} we present an initial result that defies the hope that an offline data structure can significantly speedup covering, thus motivating the study of SDC. In }Section~\ref{subsec:summary} we summarize the main technical results achieved by this paper.

\subsection{Problem Definition}
\label{subsec:defn}

We now formally define the SDC problem.

\begin{definition}[SDC]
  Given an offline input consisting of a set system $(\X,\I)$ with $n$
  \emph{elements} (a.k.a \emph{items}) $\X$ and $m$ \emph{sets} $\I$, and an integer $k\geq
  1$, devise a {\em coverage oracle} such that given a dynamic
  \emph{query} $Q\subseteq \X$, the oracle finds a $\J\subseteq \I$ such that
  $|\J|\leq k$ and $(\bigcup_{S\in \J} S)\bigcap Q$ is 
  as large as possible.
\end{definition}

\begin{definition}[Coverage Oracle]
  A {\em Coverage Oracle}  for  SDC
  consists of two stages:
\begin{enumerate}
\item {\bf Static Stage:} Given integers $m$,$n$,$k$, and set
  system $(\X,\I)$ with $|\X|=n$ and $|\I|=m$,   build a
  datastructure ${\cal D}$.
\item {\bf Dynamic Stage:} Given a a dynamic \emph{query} $Q \sse \X$,
  use ${\cal D}$ to return $\J\subseteq \I$ with $|\J| \leq k$ as a solution to SDC.
\end{enumerate}
\end{definition}

\noindent Note that our two constraints on a solution for SDC are
illustrated by the two stages above. (1) We are interested in building
an offline data structure ${\cal D}$, and only use ${\cal D}$ to
answer queries. Typically, we want to maintain a {\em small} data
structure, certainly $o(mn)$, and maybe as little as
$O((m+n)polylog(mn))$ or even $O(m+n)$. Therefore, we cannot store the
entire set system. (2) Unlike the traditional max-coverage problem
where the entire set of elements $\X$ need to be covered, in SDC we
are given queries dynamically. Therefore, we want a coverage oracle that returns good solutions for all queries.

Given the space limitation of SDC, we cannot hope to exactly solve
SDC (for all dynamic input queries). The goal of this paper is to
explore {\em approximate} solutions for SDC, given a specific space
constraint on the offline data structure ${\cal D}$. We define the
\emph{approximation ratio} of an oracle as the worst-case, taken over
all inputs, of the ratio between the coverage of $Q$ by the optimal
solution and the coverage of $Q$ by the output of the oracle. We allow
the approximation ratio to be a function of $n$, $m$, and $k$, and
denote it by $\alpha(n,m,k)$. 

More precisely, given a coverage oracle $\A$, if on inputs
$k,\X,\I,Q$ (where implicitly $n=|\X|$ and $m=|\I|$) the oracle $\A$
returns $\J \sse \I$, we denote the size of the coverage as
$\A(k,\X,\I,Q) := |(\union_{S \in \J} S) \intersect Q|$. Similarly, we
denote the coverage of the optimal solution by $OPT(k,\X,\I,Q) :=
\max\{|(\union_{S \in \J^*} S) \intersect Q|: \J^* \sse \I, |\J^*| \leq
k\}$.  We then express the \emph{approximation ratio} $\alpha(n,m,k)$ as
follows.

\[ \alpha(n,m,k) = \max
\frac{OPT(k,\X,\I,Q)}{\A(k,\X,\I,Q)} \]

Where the maximum above is taken over set systems $(\X,\I)$ with $|\X|
= n$ and $|\I|=m$, and queries $Q \sse \X$.

We will also be concerned with \emph{randomized} coverage oracles. Note that, when we devise randomized coverage oracle,
we use randomization only in the static stage; i.e. in the
construction of the datastructure. We then let the \emph{expected approximation ratio}
be the worst case \emph{expected} performance of the oracle as
compared to the optimal solution.

\begin{align} \label{eq:approx}
  \alpha(n,m,k) = \max \left(\Ex\left[\frac{OPT(k,\X,\I,Q)}{\A(k,\X,\I,Q)}\right]\right)
\end{align}

The expectation in the above expression is over the random coins
flipped by the static stage of the oracle, and the maximization is over $\X,\I,Q$ as
before. We elaborate on this benchmark in Section \ref{sec:general_ub}.

We study the space-approximation tradeoff; i.e., how the (expected)
approximation ratio improves as the amount of space allowed for ${\cal
  D}$ is increased. In our lowerbounds, we are not specifically concerned
with the time taken to compute the datastructure or answer
queries. Therefore, our lowerbounds are purely
\emph{information-theoretic}: we calculate the amount of information
we are required to store if we are to guarantee a specific
approximation ratio, independent of computational concerns. Our
lowerbounds are particularly novel and striking in that \emph{they
  assume nothing} about the datastructure, which may be an arbitrary
sequence of bits. We establish our lowerbounds via a novel application of the probabilistic method that may be of independent interest. 

Even though we focus on space vs approximation, and not on runtime,
fortunately the coverage oracles in our upperbounds can be
implemented efficiently (both static and dynamic stage). Moreover,
using our upperbounds to trade approximation for space yields, as a
side-effect, an improvement in runtime when answering a query. In
particular, observe that if no sparsification of the data is done
up-front, then answering each query using the standard greedy
approximation algorithm for max k-cover \cite{nwf78} takes $O(mn)$
time. Our oracles, presented in Section \ref{sec:general_ub}, spends
$O(mn)$ time up-front building a data structure of size $O(b)$, where
$b$ is a parameter of the oracle between $n$ and $nm$. In the dynamic
stage, however, answering a query now takes $O(b)$, since we use the
greedy algorithm for max k-cover on a ``sparse'' set
system. Therefore, the dynamic stage becomes \emph{faster} as we
decrease size of the data structure. In fact, this increase in speed
is not restricted to an algorithmic speedup as described above. It is
likely that there will also be speedup due to architectural reasons,
since a smaller amount of data needs to be kept in memory. Therefore,
trading off approximation for space yields an incidental speedup in
runtime which bodes well for the dynamic nature of the queries.


\subsection{Summary of results}
\label{subsec:summary}

Table~\ref{table:results_summary} summarizes the main results obtained
in this paper for SDC input with $n$ elements, $m$ sets, and integer
$k\geq 1$. The lower bound in the table is for any nonnegative
constants $\delta_1,\delta_2$ not both $0$, and the randomized
upperbound is parameterized by $\epsilon$ with $0\leq \epsilon \leq
1/2$. The upper and lower bounds are developed in
Sections~\ref{sec:general_ub} and~\ref{sec:general_lb}
respectively. 

\begin{table}
\caption{Summary of results for SDC giving the approximation-ratio, the space constraint on the coverage oracle, and whether the nature of the bound: upper bound (UB) or lower bound (LB) and deterministic (Det.) or randomized (Rand.)}

\begin{center}
{\small

\begin{tabular}{|c|c|c|}
\hline
{\bf Approximation Ratio} & {\bf Storage} & {\bf Bound} \\
\hline

$O\left(\min\left(\frac{m}{k},\sqrt{\frac{n}{k}}\right)\right)$ & $\tilde{O}(n)$ & Det. UB\\

\hline

$O \left(\min\left(\frac{m^\epsilon}{\sqrt{k}},\sqrt{\frac{n}{k}}\right)\right)$ & $\tilde{O}(nm^{1-2\epsilon})$ & Rand. UB\\

\hline

$\Omega\left(\min\left(\frac{m^{\epsilon-\delta_1}}{k\sqrt{k}},\frac{n^{1/2 -
      \delta_2}}{k\sqrt{k}}\right)\right)$ & $\tilde{O}(nm^{1-2\epsilon})$ & Det. LB\\


 \hline
\end{tabular}
}
\label{table:results_summary}
\end{center}
\end{table}

\eat{
\begin{table*}
\caption{Summary of Results. The notation $\left(\alpha,\beta \right)$ corresponds to a solution that gives $\alpha$-approximation and requires $\beta$ space.}
\begin{center}

\begin{tabular}{|c|c|c|}
\hline 
\multicolumn{3}{|c|} {\textbf{Problem instance: $\left(G(n,m), k\right)$}} \\
\hline
\multicolumn{3}{|c|} {\textbf{Lower bound}} \\
\hline
 & $m=2$ & $m > 2$ \\
\hline
$k=1$ & - & $(O(nm^{1-2\epsilon}),\Omega(m^{-\epsilon}))$ \eat{$\left(\frac{1}{m^{0.1 + \epsilon}}, nm^\epsilon \right)$} randomized, $\left(\frac{1}{m}, n\right)$ determenistic\\
\hline
$k>1$ & - & $\left(\frac{k}{m}, n \right)$ determenistic\\
\hline
\hline
\multicolumn{3}{|c|} {\textbf{Upper bound}} \\
\hline
 &  $m = 2$ &  $m>2$\\
\hline
$k=1$ & $\left(\frac{3}{4}, n\right)$ & $\left(\frac{1}{m}, O(n+m)\right)$ deterministic, $\left(\frac{2^{1/\epsilon}}{m}, nm^{1-\epsilon}\right)$ randomized\\
\hline
$k>1$ & - & $\left(\frac{k}{m}, O(m+n)\right)$ determenistic\\
\hline
\end{tabular}
\label{table:results_summary}
\end{center}
\end{table*}
}
\section{Upper Bounds}
\label{sec:general_ub}

In this section, we show a coverage oracle that trades off space
and approximation ratio.  We designate a tradeoff parameter
$\epsilon$, where $0 \leq \epsilon \leq 1/2$.  For any such
$\epsilon$, we get an
$O\left(\frac{\min(m^\epsilon,\sqrt{n})}{\sqrt{k}}\right)$-approximate
coverage oracle that stores $\tilde{O}(nm^{1-2\epsilon})$ bits.
Therefore, setting a small value of $\epsilon$ achieves a better
approximation ratio, at the expense of storage space. As is common
practice, we use $\tilde{O}()$ to denote suppressing polylogarithmic
factors in $n$ and $m$; this is reasonable when the guarantees are
super-polylogarithmic, as is the case here.

The oracle we show is randomized, in the sense that the static stage
flips some random coins.  The datastructure constructed is a random
variable in the internal coin flips of the static stage of the
oracle. We measure the \emph{expected approximation ratio} (a.k.a
approximation ratio, when clear from context) of the oracle, as
defined in Equation \eqref{eq:approx}. For every fixed query $Q$
independent of the random coins used in constructing the
datastructure, this ratio is attained in expectation. In
other words, our adversarial model is that of an \emph{oblivious
  adversary}: someone trying to fool our oracle may choose any query
they like, but their choice cannot depend on knowledge of the random
choices made in constructing the datastructure.

In Section~\ref{sec:general_lb} we will see that our oracle attains a
space-approximation tradeoff that is essentially optimal when compared
with oracles that are deterministic. In other words, no deterministic
oracle can do substantially better. We leave open the questions of whether a better
randomized oracle is possible, and whether an equally good
deterministic oracle exists.

\subsection{Main Result and Roadmap}

The following theorem states the main result of this section.

\begin{theorem}\label{thm:general_ub}
  For every $\epsilon$ with $0 \leq \epsilon \leq 1/2$, there is a
  randomized coverage oracle for SDC that achieves an
  $O\left(\frac{\min(m^\epsilon,\sqrt{n})}{\sqrt{k}}\right)$ approximation and
    stores $\tilde{O}(nm^{1-2\epsilon})$ bits.
 \end{theorem}

\noindent The remainder of this section, leading up to the above result, is organized as follows. 
 Before proving Theorem \ref{thm:general_ub}, to build intuition we
 show in Section~\ref{subsec:simple} (Remark \ref{rem:det_ub}) a much simpler deterministic oracle, with a much weaker
 approximation guarantee. Then, we
 prove Theorem \,\ref{thm:general_ub} in two parts. First, in Section~\ref{subsec:part1}, we show a
 randomized coverage oracle that stores
 $\tilde{O}(nm^{1-2\epsilon})$ bits and achieves an
 $O(m^\epsilon/\sqrt{k})$ approximation in expectation. Then, in Section~\ref{subsec:part2}, we show a
 deterministic oracle that achieves a $O(\sqrt{n}/\sqrt{k})$
 approximation and stores $\tilde{O}(n)$ bits. Combining the two
 oracles into a single oracle in the obvious way yields Theorem~\ref{thm:general_ub}.
 
\subsection{Simple Deterministic Oracle} 
\label{subsec:simple}

 \begin{remark}\label{rem:det_ub}
   There is a simple deterministic oracle that attains a $m/k$
   approximation with $\tilde{O}(n)$ space. The static stage proceeds
   as follows: Given set system $(\X,\I)$, for each $i\in \X$ we
   ``remember'' one set $S \in \I$ with $i \in S$ (breaking ties
   arbitrarily). In other words, for each $S \in \I$ we define $\hat{S}
   \sse S$ such that $\set{ \hat{S} : S \in \I}$ is a partition of
   $\X$. We then store the ``sparsified'' set system $\left(\X,\hat{\I}=
   \set{\hat{S} : S \in \I}\right)$. It is clear that this can be done in
   linear time by a trivial greedy algorithm. Moreover, $(\X,\hat{\I})$
   can be stored in $\tilde{O}(n)$ space as a $n \cross m$ bipartite
   graph with $n$ edges.

   The dynamic stage is straightforward: when given a query $Q$, we
   simply return the indices of the $k$ sets in $\hat{\I}$ that
   collectively cover as much of $Q$ as possible. It is clear that
   this gives a $m/k$ approximation. Moreover, since $\hat{\I}$ is a
   partition of $\X$, it can be accomplished by a trivial greedy
   algorithm in polynomial time.
 \end{remark}

 Next we use randomization to show a much better, and much more
 involved, upperbound that trades off approximation and space.

\subsection{ An $O(m^\epsilon/\sqrt{k})$ Approximation with
  $\tilde{O}(nm^{1-2\epsilon})$ Space}

\label{subsec:part1}

   Consider the set system $(\X,\I)$, where $\X$ is the set of items and
   $\I$ is the family of sets. We assume without loss that each item is
   in some set. We define a randomized oracle for building a
   datastructure, which is a ``sparsified'' version of $(\X,\I)$. Namely,
   for every $S \in \I$ we define $\hat{S} \sse S$, and store the set
   system $\left(\X,\hat{\I} = \set{\hat{S}}_{S \in \I}\right)$. We
   require that $(\X,\hat{\I})$ can be stored in $\tilde{O}(n m^{1-2\epsilon})$
   space. We construct the datastructure in two stages, as follows.

 \begin{itemize}
   \item Label all items in $\X$ ``uncovered'' and all sets in $\I$ ``unchosen''
   \item Stage 1: While there exists an unchosen set $S\in \I$
     containing at least $\frac{n}{m^\epsilon \sqrt{k}}$ uncovered items
     \begin{itemize}
     \item Let $\hat{S}$ be the set of unchosen items in $S$.
     \item Relabel all items in $\hat{S}$ as ``covered'' and  ``significant''
     \item Relabel $S$ as ``chosen'' and ``significant''
     \end{itemize}
   
   \item Stage 2: For every remaining ``unchosen'' set $S$
     \begin{itemize}
     \item Choose $\frac{n}{m^{2 \epsilon}}$ ``uncovered'' items $\hat{S} \sse S$
       uniformly at random from the uncovered items in $S$ (if fewer
       than $\frac{n}{m^{2 \epsilon}}$ such items, then let $\hat{S}$ be
       all of them).
     \item Relabel each item in $\hat{S}$ as ``covered'' and  ``insignificant''
     \item Relabel $S$ as ``chosen'' and ``insignificant''
     \end{itemize}
   \item Label every uncovered item as ``uncovered'' and  ``insignificant''

 \end{itemize}

 When presented with a query $Q \sse \X$, we use the stored
 datastructure $(\X,\hat{\I})$ in the obvious way: namely, we find
 $\hat{S_1},\ldots,\hat{S_k} \in \hat{\I}$ maximizing $|(\union_{i=1}^k
 \hat{S_i}) \intersect Q|$, and return the name of the corresponding
 original sets $S_1,\ldots,S_k$. However, this problem cannot be
 solved exactly in polynomial time in general. Nevertheless, we can
 instead use the greedy algorithm for max-k-cover to get a constant-factor
 approximation \cite{nwf78}; this will not affect our asymptotic
 guarantee on the approximation ratio. The following two lemmas complete the proof that the above oracle achieves an $O(m^\epsilon/\sqrt{k})$ approximation with
  $\tilde{O}(nm^{1-2\epsilon})$ space.
 
 \begin{lemma}
   The datastructure $(\X,\hat{\I})$ can be stored using $\tilde{O}(n m^{1-2
     \epsilon})$ bits.
 \end{lemma}
 \begin{proof}
   We store the set system as a bipartite graph representing the
   containment relation between items and sets. To show that the
   bipartite graph can be stored in the required space, it suffices to
   show that $(\X,\hat{\I})$ is ``sparse''; namely, that
   the total number of edges $(x,\hat{S}) \in \X \times \hat{\I}$ such
   that $x \in \hat{S}$ is $O(n m^{1-2 \epsilon})$. We account for the
   edges created in stages 1 and 2 separately.
   \begin{enumerate}
   \item Every significant item is connected to a single set. This
    creates at most $n$ edges.
  \item For every insignificant set, we store at most $nm^{-2\epsilon}$
    items. This creates at most $m n m^{-2\epsilon} = n m^{1-2
      \epsilon}$ edges.
   \end{enumerate}
 \end{proof}

 \begin{lemma}
   For every query $Q$, the oracle returns sets
   $S_1,\ldots,S_k$ such that \[\Ex[|(\union_{i=1}^k S_i) \intersect
   Q|] \geq \frac{|(\union_{i=1}^kS_i^*) \intersect
     Q|}{O(m^\epsilon/\sqrt{k})}\] for any $S_1^*,\ldots,S^*_k \in \I$. 
 \end{lemma}
 Note that $S_1,\ldots,S_k$ are random variables in the internal
 coin-flips of the static stage that constructs the datastructure. The
 expectation in the statement of the lemma is over these random coins.
 
 \begin{proof}
   We fix an optimal choice for $S^*_1,\ldots,S^*_k \in \I$, and denote
   $OPT=|(\union_{i=1}^k S^*_i) \intersect Q|$. Since, by
   construction, $\hat{S} \sse S$ for all $S \in \I$, it suffices to
   show that the output of the oracle satisfies $|(\union_{i=1}^k
   \hat{S_i}) \intersect Q| \geq \frac{OPT}{O(m^\epsilon/\sqrt{k})}$
   in expectation. Moreover, since the dynamic stage algorithm finds a
   constant factor approximation to $\max\{|(\union_{i=1}^k \hat{S_i})
   \intersect Q|: \hat{S_1}, \ldots, \hat{S_k} \in \hat{\I}\}$, it is
   sufficient to show that there \emph{exists}
   $S_1,\ldots,S_k \in \I$ with $\Ex[|(\union_{i=1}^k
   \hat{S_i}) \intersect Q|] \geq \frac{OPT}{O(m^\epsilon/\sqrt{k})}$.

 We distinguish two cases, based
   on whether most of the items $(\union_{i=1}^k S^*_i) \intersect Q$
   covered by the optimal solution are in significant or insignificant
   sets. We use the ``significant'' and ``insignificant'' designation
   as used in the static stage algorithm. Moreover, we refer to
   $\hat{S} \in \hat{\I}$ as significant (insignificant, resp.) when the
   corresponding $S\in \I$ is significant (insignificant, resp.).
   \begin{enumerate}
   \item {\bf At least half of $ (\union_{i=1}^k S_i^*) \intersect Q$
       are significant items}: Notice that, by construction, there are
     at most $m^\epsilon\sqrt{k}$ significant sets in
     $\hat{\I}$. Moreover, the significant items are precisely those
     covered by the significant sets of $\hat{\I}$, and those sets form
     a partition of the significant items. Therefore, by the
     pigeonhole principle there are there are some $\hat{S_1}, \ldots,
     \hat{S_k} \in \hat{\I}$ such that $\union_{i=1}^k \hat{S_i}$
     contains at least an $\frac{k}{m^\epsilon \sqrt{k}} =
     \frac{\sqrt{k}}{m^{\epsilon}}$ fraction of the significant items
     in $(\union_{i=1}^k S_i^*) \intersect Q$. This gives the desired
     $O(m^\epsilon / \sqrt{k})$ approximation.
   \item {\bf At least half of $(\union_{i=1}^k S_i^*) \intersect Q$
       are insignificant items}: In this case, at least half the items
     $(\union_{i=1}^k S_i^*) \intersect Q$ covered by the optimal
     solution are contained in the insignificant members of
     $\{S_1^*,\ldots,S^*_k\}$.
     Recall that any insignificant set in $\I$ contains at most
     $\frac{n}{m^\epsilon \sqrt{k}}$ insignificant items. Therefore,
     the algorithm includes each element of an insignificant $S_i^*$
     in $\hat{S_i^*}$ with probability at least $\frac{n}{m^{2
         \epsilon}} / \frac{n}{m^\epsilon\sqrt{k}}$, which is at least
     $\sqrt{k}/m^{\epsilon}$. Thus, every insignificant item in
     $(\union_{i=1}^k S_i^*)$ is in $(\union_{i=1}^k \hat{S_i^*})$ with
     probability at least $\sqrt{k}/m^{\epsilon}$. This gives that the expected size of
     $(\union_{i=1}^k \hat{S_i^*}) \intersect Q$ is at least
     $\frac{OPT}{O(m^\epsilon / \sqrt{k})}$. Taking $S_i=S^*_i$ completes the proof.
   \end{enumerate}
 \end{proof}

\subsection{An $O(\sqrt{n/k})$ Approximation with $\tilde{O}(n)$ Space}
\label{subsec:part2}

This coverage oracle is similar to the one in the previous
section, though is much simpler. Moreover, it is
deterministic. Indeed, we construct the datastructure by the following
greedy algorithm that resembles the greedy algorithm for max-k-cover

 \begin{itemize}
   \item Label all items in $\X$ ``uncovered'' and all sets in $\I$ ``unchosen''
   \item While there are unchosen sets
     \begin{itemize}
     \item Find the unchosen set $S \in \I$ containing the most
       uncovered items
     \item Let $\hat{S}$ be the set of uncovered items in $S$.
     \item Relabel all items in $\hat{S}$ as ``covered''
     \item Relabel $S$ as ``chosen''
     \end{itemize}
   \end{itemize}

   Observe that $\hat{\I}$ is a partition of $\X$.  When presented with
   a query $Q \sse \X$, we use the datastructure
   $(\X,\hat{\I}=\set{\hat{S}: S \in \I})$ in the obvious way. Namely, we
   find the sets $\hat{S_1}, \ldots, \hat{S_k} \in \hat{\I}$ maximizing
   $|(\union_{i=1}^k \hat{S_i}) \intersect Q|$, and output the
   corresponding non-sparse sets $S_1,\ldots,S_k$. This can easily be
   done in polynomial time by using the obvious greedy algorithm,
   since $\hat{\I}$ is a partition of $\X$. 
   
   Note that the oracle described above is very similar to the oracle
   from Section~\ref{subsec:simple}: The dynamic stage is
   identical. The static stage, however, needs to build the partition
   using a specific greedy ordering -- as opposed to the arbitrary
   ordering used in Section \,\ref{subsec:simple}.  The following two
   Lemmas complete the proof that the oracle achieves an
   $O(\sqrt{n/k})$ approximation with $\tilde{O}(n)$ space.
   
\begin{lemma}
  The datastructure $(\X,\hat{\I})$ can be stored using $\tilde{O}(n)$ bits
\end{lemma}
\begin{proof}
  Observe that each item is contained in exactly one $\hat{S} \in
  \hat{\I}$. Therefore, the bipartite graph representing the set system
  $(\X,\hat{\I})$ has at most $n$ edges. This establishes the Lemma.
\end{proof}

\begin{lemma}
  For every query $Q$, the oracle returns sets $S_1,\ldots,S_k$
  with \[|(\union_{i=1}^k S_i)
  \intersect Q| \geq \frac{| (\union_{i=1}^k S_i^*) \intersect Q|}{O(\sqrt{n/k})}\] for any
  $S_1^*,\ldots,S^*_k \in \I$.
\end{lemma}
\begin{proof}
  Fix an optimal choice of $S_1^*,\ldots,S^*_k$, and denote
  $OPT=|(\union_{i=1}^k S_i^*) \intersect Q|$. Recall that the oracle
  finds $\hat{S_1}, \ldots, \hat{S_k} \in \hat{\I}$ maximizing $|(\union_{i=1}^k
  \hat{S}_i) \intersect Q|$, and then outputs the corresponding
  original sets $S_1,\ldots,S_k$.

  It suffices to show that there are some $\hat{S}_1,\ldots,\hat{S}_k
  \in \hat{\I}$ with $|(\union_{i=1}^k \hat{S}_i) \intersect Q| \geq
  OPT/O(\sqrt{n/k})$. We distinguish two cases, based on whether most of
  $(\union_{i=1}^k S_i^*) \intersect Q$ are in big or small sets in $\hat{\I}$.

  Recall that $\hat{\I}$ forms a partition of $\X$. We say $\hat{S} \in
  \hat{\I}$ is ``significant'' if $|\hat{S}| \geq \sqrt{n/k}$, otherwise
  $\hat{S}$ is ``insignificant''. Similarly, we say an item $i \in \X$
  is ``significant'' if it falls in a significant set in $\hat{\I}$,
  otherwise it is ``insignificant''. Notice that there are at most
  $\frac{n}{\sqrt{n/k}} =\sqrt{nk}$ significant sets.

  First, we consider the case where at least half the items in
  $(\union_{i=1}^k) S_i^* \intersect Q$ are significant. Since there
  at most $\sqrt{nk}$ significant sets in $\hat{\I}$, by the pigeonhole
  principle there are $k$ of them that collectively cover a
  $k/\sqrt{nk} = \sqrt{k/n}$ fraction of all significant items in
  $(\union_{i=1}^k S_i^*) \intersect
  Q$. This would guarantee the $O(\sqrt{n/k})$ approximation, as needed.

  Next, we consider the case where at least half of $(\union_{i=1}^k
  S_i^*) \intersect Q$ are insignificant. By examining the greedy
  algorithm of the static stage, it is easy to see that each $S \in
  \I$ contains at most $\sqrt{n/k}$ insignificant items. Therefore,
  there are at most $k \cdot \sqrt{n/k} = \sqrt{nk}$ insignificant
  items in $(\union_{i=1}^k S_i^*)$. Therefore we deduce that
  $OPT=|(\union_i S_i^*) \intersect Q| \leq 2 \sqrt{kn}$. Since the
  optimal covers $O(\sqrt{kn})$ items in $Q$, it suffices for a
  $O(\sqrt{n/k})$ approximation to show that there are
  $\hat{S}_1,\ldots,\hat{S}_k \in \hat{\I}$ that collectively cover
  $k$ items of $Q$. It is easy to see that this is indeed the case,
  since $\hat{\I}$ is a partition of $\X$. This completes the proof.
\end{proof}

\section{Lower Bounds}
\label{sec:general_lb}

This section develops lower bounds for the SDC problem. We consider
deterministic oracles that store a data\-structure of size $b(n,m,k)$
for set systems with $n$ items, $m$ sets, maximum number of allowed
sets $k$. Moreover, we assume that $n \leq b(n,m,k) \leq nm$, since no
nontrivial positive result is possible when $b(n,m,k) = o(n)$, and a
perfect approximation ratio of $1$ is possible when $b(n,m,k) =\Omega(nm)$. 

\subsection{Main Result and Roadmap}

The main result of this section is stated in the following theorem, which says that our
randomized oracle in the previous section achieves a space-approximation
tradeoff that essentially matches the best possible for any
deterministic oracle. 

\begin{theorem}\label{thm:general_lb}
  Consider any deterministic oracle that stores a datastructure of
  size at most $b(n,m,k)$ bits, where $n \leq b(n,m,k) \leq nm$. Let
  $\epsilon(n,m,k)$ be such that $b(n,m,k) = n m^{1-2
    \epsilon(n,m,k)}$. When $m^{\epsilon(n,m,k)} \leq \sqrt{n}$, the
  oracle does not attain an approximation ratio of $O(\frac{m^{\epsilon(n,m,k)
    - \delta}}{k\sqrt{k}})$ for any constant $\delta > 0$. Moreover, when
  $\sqrt{n} \leq m^{\epsilon(n,m,k)}$ the oracle does not attain an
  approximation ratio of $O(\frac{n^{1/2 - \delta}}{k\sqrt{k}})$ for any $\delta >0$.
 \end{theorem}

\noindent   The proof of the theorem above is somewhat involved.  
   Therefore, to simplify the
   presentation we prove in Section \ref{sec:general_lb_simplified} a
   slight simplification of Theorem \ref{thm:general_lb} that
   captures all the main ideas: Our simplification sets $k=1$, and proves the $O(\frac{m^{\epsilon(n,m,k)
    - \delta}}{k\sqrt{k}})$ approximation ratio, for $m^{\epsilon(n,m,k)} \leq \sqrt{n}$. Then, in Section~\ref{sec:general_lb_generalize1} we prove the approximation ratio for the case of $\sqrt{n} \leq m^{\epsilon(n,m,k)}$, still maintaining $k=1$. Finally, in Section~ \ref{sec:general_lb_generalize2}, we demonstrate how to modify our proofs for any $k$, yielding Theorem \ref{thm:general_lb}.

We fix $\delta >0$. For the remainder of the section, we use $b$ and $\epsilon$ as shorthand for
   $b(n,m,k)$ and $\epsilon(n,m,k)$, respectively. We let $\alpha(n,m,k)$ be
   the approximation ratio of the oracle, and use $\alpha$ as
   shorthand. Observe that $0 \leq \epsilon \leq 1/2$.

\subsection{Proof of a Simpler Lowerbound}
\label{sec:general_lb_simplified}
We simplify Theorem \ref{thm:general_lb} by assuming $k=1$ and
$m^\epsilon \leq \sqrt{n}$. The result is the following proposition,
stated using the shorthand notation described above.
\begin{prop}\label{prop:general_lb_simplified}
  Fix $k=1$ and parameter $\epsilon$ with $0 \leq \epsilon \leq
  1/2$. Assume $m^\epsilon \leq \sqrt{n}$. Consider any deterministic
  oracle that stores a datastructure of size at most
  $b=nm^{1-2\epsilon}$ bits. The oracle does not attain an
  approximation ratio of $O(m^{\epsilon - \delta})$ for any constant
  $\delta > 0$.
\end{prop}

We assume the approximation ratio $\alpha$ attained by the oracle is
$O(m^{\epsilon - \delta})$ and derive a contradiction. The proof uses
the probabilistic method (see \cite{alonspencer}). We begin by
defining a distribution on set systems, and then go on to show that
this distribution ``fools'' a small coverage oracle with positive
probability.

\subsubsection{Defining a Distribution $D$ on Set Systems}
\label{subsec:lb_distribution}
   We will show that there is a set system $(\X,\I)$ and a query $Q$
   that forces the algorithm to output a set $S\in \I$ that is not
   within $\alpha$ from optimal. We use the probabilistic
   method. Namely, we exhibit a distribution $D$ over set systems
   $(\X,\I)$ such that, for every deterministic oracle storing a
   datastructure of size $b$, there exists with non-zero probability a
   query $Q$ for which the oracle outputs a set of approximation worse
   than $\alpha$. To show this, we draw two set systems i.i.d from
   $D$, and show that with non-zero probability both the following
   hold: the two set systems are not distinguished by the
   coverage oracle, and moreover there exists a query $Q$ that
   requires that the algorithm return different answers for the two
   set systems for a $O(m^{\epsilon - \delta})$
   approximation. 

   We define $D$ as follows.  Given the ground set
   $\X=\set{1,\ldots,n}$, we let $\I=\set{A_i}_{i=1}^m$ and draw
   $A_1,\ldots,A_m$ i.i.d as follows: We let $A_i$ be a subset of $\X$
   of size $nm^{-\epsilon}$ drawn uniformly at random.

\subsubsection{Sampling twice from $D$ and collisions}
\label{subsec:lb_sampling_collisions}
   Next, we draw two set systems $(\X,\I=\set{A_i}_{i=1}^m)$ and
   $(\X,\I'=\set{A'_i}_{i=1}^m)$ i.i.d from $D$, as discussed
   above. First, we lowerbound the probability that $(\X,\I)$ and
   $(\X,\I')$ are not distinguished by the coverage oracle. We call
   such an occurence a ``Collision''.

   \begin{lemma}\label{lem:collision}
     The probability that the same datastructure is stored for $(\X,\I)$
     and $(\X,\I')$ is at least $2^{-b} $.
   \end{lemma}
   \begin{proof}
     There are $2^b$ possible datastructures. Let $p_i$ denote the
     probability that, when presented with random $(\X,\I) \sim D$,
     the oracle stores the $i$'th datastructure. We can write this
     probability of ``collision'' of the two i.i.d samples $(\X,\I)$ and
     $(\X,\I')$ as $\sum_{i=1}^{2^b} p_i^2 $.  However, since $\sum_i
     p_i = 1$, this expression is minimized when $p_i = 2^{-b}$ for
     all $i$. Plugging into the above expression gives a lowerbound of
     $2^{-b}$, as required.
   \end{proof}

\subsubsection{Fooling Queries and Candidates}
\label{subsec:lb_fooling}
Next, we lowerbound the probability that a query $Q$ exists requiring
two different answers for $(\X,\I)$ and $(\X,\I')$ in order to get the
desired $\alpha = O(m^{\epsilon - \delta})$ approximation. We call such a query $Q$ a
\emph{fooling query}. We define a set of queries that are
``candidates'' for being a fooling query: A set $Q \sse \X$ is called a
\emph{candidate query} if $Q= A_i \union A'_{i'}$ for some $i \neq
i'$. In other words, a query is a candidate if it is the union of a set
from $(\X,\I)$ and a set from $(\X,\I')$ with different indices. 

Ideally, candidate $Q=A_i \union A'_{i'}$ would be a fooling query by forcing the oracle
to output $i$ for $(\X,\I)$ and $i'$ for $(\X,\I')$ in order to guarantee
the desired approximation. However, this need not be the case:
consider for instance the case when, for some $j\neq i,i'$, both $A_j$
and $A'_j$ have large intersection with $Q$, making it ok to output
$j$ for both. We will show that the probability that none of the
candidate queries is a fooling query is strictly less
than $2^{-b}$ when $n$ and $m$ are sufficiently large. Doing so would
complete the proof: collision occurs with probability $\geq 2^{-b}$,
and a fooling query exists with probability $> 1-2^{-b}$, and
therefore both occur simultaneously with positive probability. This
would yield the desired contradiction.

\subsubsection{The Probability that None of the Candidates is Fooling
  is Small}
\label{subsec:lb_fooling_likely}

We now upperbound the probability that none of the candidates is a
fooling query. Observe that if candidate $Q=A_i \union A'_{i'}$ is not
a fooling query, then there exists $A \in \I \union \I'
\sm\{A_i,A'_{i'}\}$ with $|A \intersect Q| \geq nm^{- \epsilon} /
\alpha$. Therefore one of the following must be true:
   \begin{enumerate}
   \item There exists $A \in \I \union \I' \sm\{A_i,A'_{i'}\}$ with $|A
     \intersect A_i| \geq nm^{- \epsilon} / 2 \alpha = \Omega(n
     m^{-2\epsilon + \delta})$.
   \item There exists $A \in \I \union \I' \sm\{A_i,A'_{i'}\}$ with $|A
     \intersect A'_{i'}| \geq nm^{- \epsilon} / 2 \alpha = \Omega(n
     m^{-2\epsilon + \delta})$.

   \end{enumerate}

   Therefore, if none of the candidates were fooling
   queries, then there are many  ``pairs'' of sets in $\I \union \I'$
   that have an intersection substantially larger than the expected
   size of $n m^{-2 \epsilon}$. This seems very unlikely. Indeed, the
   remainder of this proof will demonstrate just that.

   If none of the candidates are fooling queries, then by examining
   (1) and (2) above we deduce the following. There exists
   \footnote{Consider constructing $P$ as follows: For candidate query
     $Q= A_1 \union A'_2$, find the set in $\I \union \I' \sm
     \set{A_1,A'_2}$ with a large intersection with one of $A_1$ or
     $A'_2$ as in (1) or (2). Say for instance we find that $A_7$ has
     a large intersection with $A_1$. We include $(A_1,A_7)$ in $P$,
     mark both $A_1$ and $A_7$ as ``touched'', and designate $A_1$ a
     ``left'' node and $A_7$ a ``right'' node. Then, we repeat the
     process with some candidate $Q'= A_i \union A'_{i'}$ for some
     ``untouched'' $A_i$ and $A'_{i'}$. We keep repeating until there
     are no such candidates. Throughout this greedy process, we mark
     at most two members of $\I \union \I'$ as ``touched'' for every pair we include in
     $P$. Note that some $A_i$ may be ``touched'' more than once. As
     long as there are at least 2 untouched sets in each of $\I$ and
     $\I'$, the algorithm may continue.}
   a set of pairs $P \sse (\I \union \I') \cross (\I \union \I')$ such
   that:

   \begin{enumerate}
   \item $|P| \geq m-2 = \Omega(m)$
   \item The undirected graph with nodes $\I \union \I'$ and edges $P$ is
     bipartite. Moreover, every node in the left part has degree at
     most $1$. Thus $P$ is acyclic.
   \item If $(B,C) \in P$ then $| B \intersect C| \geq \Omega(n
     m^{-2\epsilon + \delta} )$
   \end{enumerate}

   We now proceed to bound the probability of existence of such a $P$,
   and in the process also bound the probability that none of the
   candidate queries are fooling. Recall that members of $\I \union \I'$
   are drawn i.i.d from the uniform distribution on subsets of $\X$ of
   size $nm^{-\epsilon }$. For every pair $(B,C) \in \I \union \I'$, we
   let $\R(B,C) = |B \intersect C|$ denote the size of their
   intersection. It is easy to see the random variables
   $\set{\R(B,C)}_{B,C \in \I \union \I'}$ are pairwise
   independent. Therefore, any acyclic set of pairs is mutually
   independent, by basic probability theory. Thus, if we fix a particular $P$
   satisfying (1) and (2), the probability that $P$ satisfies
   condition (3) is at most

\[ \prod_{(B,C) \in P} \Pr[ \R(B,C) \geq \Omega(n
     m^{-2\epsilon + \delta} )]  \] 

     We now want to estimate the probability that the intersection of
     $B$ and $C$ is a factor $\Omega(m^\delta)$ more than its
     expectation of $nm^{-2 \epsilon}$. Therefore, we consider an
     indicator random variable $Y_i$ for each $i \in \X$, designating
     wheter $i \in B \cap C$. If $Y_i$ were independent, we could use
     Chernoff bounds to bound the probability that $\R(B,C)$ is
     large. Fortunately, it is easy to see that the $Y_i$'s are
     negatively-correlated: i.e., for any $L \sse \set{1,\ldots,n}$,
     we have $\Pr[ \bigwedge_{i \in L} Y_{i}=1] \leq \prod_{i \in L}
     \Pr[Y_i = 1] $ . Therefore, by the result of \cite{panconesi}, if
     we ``pretend'' that they are independent by approximating their
     joint-distribution by i.i.d bernoulli random variables, we can
     still use Chernoff Bounds to bound the upper-tail
     probability. Therefore, using Chernoff bounds\footnote{We use the
       following version of the Chernoff Bound: Let $X_1,\ldots,X_n$ be independent
       bernoulli random variables, and let $X=\sum_i X_i$. If $\Ex[X]
       = \mu$ and $\Delta > 2e-1$, then $\Pr[ X > (1+\Delta) \mu ]
       \leq 2^{-\Delta \mu}$.} we deduce that the probability that the
     intersection of $B$ and $C$ is a factor $\Omega(m^\delta)$ more
     than the expectation of $nm^{-2 \epsilon}$ is at most $2^{-
       (\Omega(m^\delta) - 1)nm^{-2 \epsilon}} \leq 2^{- \Omega(nm^{-
         2 \epsilon + \delta})}$. Therefore, the probability that the
     fixed $P$ satisfies condition (3) is at most
     \begin{align*}
       \prod_{(B,C) \in P} 2^{- \Omega(nm^{- 2 \epsilon + \delta})}
     &\leq (2^{- \Omega(nm^{- 2 \epsilon + \delta})})^{|P|} \\
     &\leq 2^{- \Omega(nm^{1- 2 \epsilon + \delta})}
     \end{align*}

     Now, we can sum over all possible choices for $P$ satisfying (1)
     and (2) to get a bound on the existence of a $P$ satisfying (1),
     (2) and (3). It is easy to see that there are at most $m^m$
     choices for $P$ that satify (1) and (2). Using the union bound,
     we get the following bound on the existence of such a $P$.
     \begin{align*}
       m^m \cdot 2^{- \Omega(nm^{1- 2 \epsilon + \delta})} &\leq 2^{ m
         \log m- \Omega(nm^{1- 2 \epsilon + \delta})}\\
         &\leq 2^{-\Omega(nm^{1- 2 \epsilon + \delta})}
     \end{align*}

     Where the last inequality follows by simple algebraic
     manipulation from our assumption that $m^\epsilon \leq \sqrt{n}$
     and $\delta >0$, when $n$ and $m$ are sufficiently large.  Recall
     that, by our previous discussion, this expression also
     upperbounds the probability that none of the candidate queries
     are fooling queries. But, when $n$ and $m$ are
     sufficiently large, this is strictly smaller than $2^{-b} =
     2^{-nm^{1-2 \epsilon}}$. Thus, by our previous discussion, this
     completes the proof of Proposition
     \ref{prop:general_lb_simplified}.

\subsection{Modifying the proof for the case $\sqrt{n} \leq
  m^\epsilon$}
\label{sec:general_lb_generalize1}
We maintain the assumption that $k=1$, and show how to modify the
proof of Proposition \ref{prop:general_lb_simplified} for the case
when $\sqrt{n} \leq m^\epsilon$.
\begin{prop}\label{prop:general_lb_generalized1}
  Fix $k=1$ and parameter $\epsilon$ with $0 \leq \epsilon \leq
  1/2$. Assume $\sqrt{n}\leq m^\epsilon$. Consider any deterministic
  oracle that stores a datastructure of size at most
  $b=nm^{1-2\epsilon}$ bits. The oracle does not attain an
  approximation ratio of $O(n^{1/2 - \delta})$ for any constant
  $\delta > 0$.
\end{prop}

Instead of replicating almost the entire proof
of Proposition \ref{prop:general_lb_simplified}, we instead point out the key
changes necessary to yield a proof of
\ref{prop:general_lb_generalized1} and leave the rest as an easy
excercise for the reader.

The proof proceeds almost identically to the proof of Proposition
\ref{prop:general_lb_simplified}, with the following main changes:

\begin{itemize}
\item {\bf Modifications to Section \,\ref{subsec:lb_distribution}}: When defining $D$, we
  let each $A_i$ be a subset of $\X$ of size $\sqrt{n}$ instead of
  $nm^{-\epsilon}$.
\item We perform similar calculations throughout, accomodating the
  above modification to the size of $A_i$.
\item {\bf Modifications to Section \,\ref{subsec:lb_fooling_likely}}: We eventually
  arrive at an upper bound of $2^{-mn^\delta}$ on the probability that
  none of the candidate queries are fooling. Using the assumption
  $m^\epsilon \geq \sqrt{n}$ and the fact that $b=nm^{1-2\epsilon}$, a
  simple algebraic manipulation shows that this bound is stricly less
  than $2^{-b}$. This completes the proof, as before.
\end{itemize}

\subsection{Modifying the proof for  arbitrary $k$}
\label{sec:general_lb_generalize2}

In this section, we generalize Proposition \,\ref{prop:general_lb_simplified} to
arbitrary $k$. The generalization of Proposition
\,\ref{prop:general_lb_generalized1} to arbitrary $k$ is essentially
identical, and therefore we leave it as an exercise for the reader. We
now state the generalization of Proposition
\,\ref{prop:general_lb_simplified} to arbitrary $k$.

\begin{prop}\label{prop:general_lb_generalized2}
  Let parameter $\epsilon$ be such that $0 \leq \epsilon \leq 1/2$. Assume
  $m^\epsilon \leq \sqrt{n}$. Consider any deterministic oracle that
  stores a datastructure of size at most $b=nm^{1-2\epsilon}$
  bits. The oracle does not attain an approximation ratio of
  $O(\frac{m^{\epsilon - \delta}}{k \sqrt{k}})$ for any constant $\delta > 0$.
\end{prop}

The proof of Proposition \ref{prop:general_lb_generalized2} follows
the outline of the proof of Proposition
\,\ref{prop:general_lb_simplified}. The necessary modifications to the
proof of Proposition \,\ref{prop:general_lb_simplified} are as follows:

\begin{itemize}
\item {\bf Modifications to Section \,\ref{subsec:lb_distribution}}: We define
  distribution $D$ as before, except that we let each $A_i$ be a
  subset of $\X$ of size $nm^{-\epsilon} \sqrt{k}$.
\item {\bf Modifications to Section \,\ref{subsec:lb_sampling_collisions}}: Instead of
  sampling from $D$ twice, we sample $2k+1$ times to get set systems
  $(\X,\I^1), (\X,\I^2),\ldots,(\X,\I^{2k+1})$. This changes the probability
  of collision of Lemma \,\ref{lem:collision} to $2^{-2kb} $. Here,
  collision means that all $2k+1$ samples from $D$ are stored as the
  same datastructure by the static stage of the oracle.
\item {\bf Modifications to Section \,\ref{subsec:lb_fooling}}: We now define a
  \emph{fooling query} analogously for general $k$: A query $Q$ is
  fooling if there is no single index $i$ such that returning the
  $i$'th set gives a good approximation for all the set systems
  $(\X,\I^1)$, \ldots $(\X,\I^{2k+1})$.

  Moreover, we analogously define \emph{candidate queries}: We use
  $A^a_b$ to denote the $b$'th set in set system $(\X,\I^a)$. We say $Q
  \sse \X$ is a candidate if $Q=A^{\ell_1}_{i_1} \union
  A^{\ell_2}_{i_2} \union \ldots \union A^{\ell_{k+1}}_{i_{k+1}}$,
  where indices $\ell_1,\ldots,\ell_{k+1}$ are distinct, and indices
  $i_1,\ldots,i_{k+1}$ are distinct. In other words, $Q$ is a fooling
  query if it is the union of $k+1$ sets from $k+1$ distinct set
  systems and $k+1$ distinct indices in those set systems.
\item {\bf Modifications to Section \,\ref{subsec:lb_fooling_likely}}: Similarly, if a
  candidate $Q=A^{\ell_1}_{i_1} \union \ldots \union
  A^{\ell_{k+1}}_{i_{k+1}}$ is not a fooling query, then there is some
  $A \in (\I^1 \union \ldots \I^{2k+1}) \sm \set{A^{\ell_j}_{i_j}}_j$ with $|A \intersect Q|\geq
  nm^{-\epsilon}/\alpha$. 
  Therefore, for one of the components $A^{\ell_j}_{i_j}$ of $Q$ we
  have that $|A \intersect A^{\ell_j}_{i_j}| \geq nm^{-\epsilon}/k
  \alpha$. Plugging in the approximation ratio $\alpha=m^{\epsilon -
    \delta}/k\sqrt{k}$ we have that $|A \intersect A^{\ell_j}_{i_j}|
  \geq nm^{-2 \epsilon +\delta} \sqrt{k}$. It is not too hard to see
  that we can construct $P$ similarly with

   \begin{enumerate}
   \item $|P| \geq k(m-k) = \Omega(km)$.\footnote{This is not true when $k$
         is almost equal to $m$. However, the theorem becomes
         trivially true when
         $k>m^{1/6}$, so we can without loss assume that $k$ is not
         too large.}
   \item The undirected graph with nodes $\I \union \I'$ and edges $P$ is
     bipartite. Moreover, every node in the left part has degree at
     most $1$. Thus $P$ is acyclic.
   \item If $(B,C) \in P$ then $| B \intersect C| \geq \Omega(n
     m^{-2\epsilon + \delta} \sqrt{k} )$
   \end{enumerate}

   Continuing with the remaining calculations in this section almost
   identically gives a bound of
   $2^{-\Omega(knm^{1-2\epsilon+\delta})}$ on the probability of
   existance of a fixed $P$. The number of such $P$ is at most
   $(km)^{km}$, therefore a similar calculation gives a bound
   of \[2^{-\Omega(knm^{1-2\epsilon+\delta})} =
   2^{-\Omega(kbm^\delta)}\] on the \emph{existence} of any such $P$.
   As before, this completes the proof.
\end{itemize}

\section{Conclusions and Future Work}
\label{sec:conclusions}

This paper introduced and studied a fundamental problem, called SDC, arising in many large-scale Web applications. A summary of results obtained by the paper appear in Table~\ref{table:results_summary} (Section~\ref{subsec:summary}). The main specific open question that arises is whether there is a deterministic oracle that is as good as the randomized oracle proposed in Section~\ref{sec:general_ub}. More generally, a detailed analysis of practical subclasses of SDC seems to hold promise.

\section*{Acknowledgements}
We thank Philip Bohannon, Hector Garcia-Molina, Ashwin Machavanajhhaala, Tim Roughgarden, and Elad Verbin for insightful discussions. 


{\small
\bibliographystyle{plain}
\bibliography{covering.bib}
}

\end{document}